\documentclass[
11pt]{article}
\usepackage{times}
\usepackage{amssymb,amsfonts,latexsym,stmaryrd}
\usepackage{QED}
\usepackage[PostScript=dvips]{diagrams}
\usepackage{euscript}

\newtheorem{theorem}{Theorem}[section]

\newtheorem{proposition}[theorem]{Proposition}

\newcommand{\Imp}{\; \Rightarrow \;}

\newcommand{\IFF}{\;\; \Longleftrightarrow \;\;}

\newcommand{\id}{\mathsf{id}}
\newcommand{\Set}{\mathbf{Set}}
\newcommand{\sSet}{\mathbf{sSet}}
\newcommand{\Chu}{\mathbf{Chu}}
\newcommand{\Chus}{\mathbf{sChu}}
\newcommand{\rarr}{\rightarrow}
\newcommand{\lrarr}{\longrightarrow}
\newcommand{\pow}{\mathcal{P}}
\newcommand{\CC}{\EuScript{C}}

\newcommand{\HH}{\mathcal{H}}
\newcommand{\KK}{\mathcal{K}}
\newcommand{\LL}{\mathsf{L}}
\newcommand{\PP}{\mathsf{P}}
\newcommand{\UU}{\mathcal{U}}
\newcommand{\Ho}{\HH_{\circ}}

\newcommand{\Uo}{U_{\circ}}
\newcommand{\Ko}{\KK_{\circ}}
\newcommand{\ip}[2]{\langle #1 \mid #2 \rangle}
\newcommand{\norm}[1]{\| #1  \|}
\newcommand{\ray}[1]{\bar{#1}}
\newcommand{\eH}{e_{\HH}}
\newcommand{\eK}{e_{\KK}}
\newcommand{\beH}{\bar{e}_{\HH}}
\newcommand{\beK}{\bar{e}_{\KK}}
\newcommand{\Complex}{\mathbb{C}}

\newcommand{\GG}{\mathbb{G}}

\newcommand{\ie}{\textit{i.e.}~}

\newcommand{\PSymm}{\mathbf{\PP SymmH}}

\newcommand{\Two}{\mathbf{2}}

\newcommand{\NChu}{\mathbf{NChu}}
\newcommand{\FCo}{F_K{-}\mathbf{Coalg}}
\newcommand{\Pfn}{\mathbf{Pfn}}
\newcommand{\One}{\{ 0 \}}

\newcommand{\gQ}{\gamma_Q}
\newcommand{\aH}{a_{\HH}}
\newcommand{\aK}{a_{\KK}}
\newcommand{\baH}{\bar{a}_{\HH}}
\newcommand{\baK}{\bar{a}_{\KK}}
\newcommand{\FCoalg}{F{-}\mathbf{Coalg}}
\newcommand{\Prob}{\mathsf{Prob}}
\newcommand{\CAT}{\mathbf{CAT}}
\newcommand{\CHU}{\mathsf{Chu}}
\newcommand{\schu}{\mathsf{sChu}}
\newcommand{\II}{\mathbb{I}}
\newcommand{\Nat}{\mathbb{N}}
\newcommand{\IndF}{\mathsf{F}}
\newcommand{\sind}{\mathsf{sF}}
\newcommand{\IndS}{\mathsf{S}}
\newcommand{\TT}{\mathsf{T}}
\newcommand{\EE}{\mathsf{E}}
\newcommand{\Coalg}{\mathbf{Coalg}}
\newcommand{\fsem}[1]{\llbracket #1 \rrbracket}
\newcommand{\bis}{\sim_{b}}
\newcommand{\proj}{\sim_{p}}
\newcommand{\Ug}{(\UU_{Q}, \gQ)}

\newarrow{inc}C--->
\newarrow{mon}>--->
\newarrow{epi}----{>>}
\newarrow{rep}>---{>>}

\begin{document}

\title{Coalgebras, Chu Spaces,  and
Representations of Physical Systems}
\author{Samson Abramsky\\
Oxford University Computing Laboratory}

\maketitle

\begin{abstract}
We revisit our earlier work on the representation of quantum systems as Chu spaces, and investigate the use of coalgebra as an alternative framework. On the one hand, coalgebras allow the dynamics of repeated measurement to be captured, and provide mathematical tools such as final coalgebras, bisimulation and coalgebraic logic. However, the standard coalgebraic framework does not accommodate contravariance, and is too rigid to allow physical symmetries to be represented.
We introduce a fibrational structure on coalgebras in which contravariance is represented by indexing. We use this structure to give a universal semantics for quantum systems based on a final coalgebra construction. We characterize equality in this semantics as projective equivalence. 
We also define an analogous indexed structure for Chu spaces, and use this to obtain a novel categorical description of the category of Chu spaces. We use the indexed structures of Chu spaces and coalgebras over a common base to define a truncation functor from coalgebras to Chu spaces.
This truncation functor is used to lift the full and faithful representation of the groupoid of physical symmetries on Hilbert spaces into Chu spaces, obtained in our previous work, to the coalgebraic semantics.
\end{abstract}


\section{Introduction}

Chu spaces and universal coalgebra are two general formalisms for systems modelling in a broad sense. Both have been studied quite extensively in Computer Science over the past couple of decades. Recently, we showed how quantum systems with their symmetries have a full and faithful representation as Chu spaces \cite{btm}. We had in fact originally intended to use coalgebras as the vehicle for this work. This did not prove satisfactory, for reasons which will be explained later. But coalgebras have many features which make them promising for studies of this kind. Moreover, as we shall show, the problems which arise can in fact be overcome to a considerable degree, in a fashion which brings to light some interesting and novel aspects of these two well-studied models, and in particular of \emph{the relationships between them} --- which have, to the best of our knowledge, not been studied at all previously. 

The purpose of the present paper is thus to develop some systematic connections and contrasts between Chu spaces and coalgebras, the modelling issues which arise, what can be done to resolve them, and which problems remain outstanding.

The main results of our investigations can be summarized as follows:
\begin{itemize}
\item Firstly, at the general level, we look at the comparative strengths and weaknesses of the two formalisms. On our analysis, the key feature that Chu spaces have and coalgebras lack is \emph{contravariance}; the key feature which coalgebras have and Chu spaces lack is \emph{extension in time}. There are some interesting secondary issues as well, notably \emph{symmetry vs.~rigidity}.

\item Formally, we introduce an indexed structure for coalgebras to compensate for the lack of contravariance, and show how this can be used to represent a wide class of physical systems in coalgebraic terms. In particular, we show how a \emph{universal model for quantum systems} can be constructed as a \emph{final coalgebra}. This opens the way to the use of methods such as \emph{coalgebraic logic} in the study of physical systems. It also suggests how coalgebra can mediate between \emph{ontic} and \emph{epistemic} views of the states of physical systems.

\item We also define an analogous indexed structure for Chu spaces, and use this to obtain a novel categorical description of the category of Chu spaces. We use the indexed structures of Chu spaces and coalgebras over a common base to define a \emph{truncation functor} from coalgebras to Chu spaces.

\item We use this truncation functor to lift the full and faithful representation of the groupoid of physical symmetries on Hilbert spaces into Chu spaces, obtained in \cite{btm}, to the coalgebraic semantics.
\end{itemize}

The further contents of the paper are organized as follows. In Section~2 we review some background on Chu spaces and coalgebras. In section~3 we make a first comparison of Chu spaces and coalgebras. Then in Section~4 we discuss the modelling issues, the problems which arise, and the strengths and weaknesses of the two approaches. In Section~5 we develop the technical material on indexed structure for coalgebras. A similar development for  Chu spaces is carried out in Section~6, and the  truncation functor is defined. In Section~7 we show how a universal model for quantum systems can be constructed as a final coalgebra; equality in the coalgebraic semantics is characterized as projective equivalence, and the representation theorem for the symmetry groupoid on Hilbert spaces is lifted from Chu spaces to the coalgebraic category. Section~8 outlines the general scheme of `bivariant coalgebra' underlying our approach. 

\section{Background}

\subsection{Coalgebra}
Coalgebra has proved to be a powerful and flexible tool for modelling a wide range of systems.
We shall give a very brief introduction. Further details may be found e.g. in the excellent presentation in \cite{DBLP:journals/tcs/Rutten00}.

Category theory allows us to \emph{dualize}  algebras to obtain a notion of \emph{coalgebras of an endofunctor}. However, while algebras abstract a familiar set of notions, coalgebras open up a new and rather unexpected territory, and provides an effective abstraction and mathematical theory for a central class of computational phenomena:
\begin{itemize}
\item Programming over \emph{infinite data structures}: streams, infinite trees, etc.

\item A novel notion of \emph{coinduction}.

\item Modelling \emph{state-based computations} of all kinds.

\item The key notion of \emph{bisimulation equivalence} between processes.

\item A general \emph{coalgebraic logic} can be read off from the functor, and used to specify and reason about properties of systems.
\end{itemize}

Let $F : \CC \rarr \CC$ be a functor.
An \emph{$F$-coalgebra} is a pair $(A, \alpha)$ where $A$ is an object of $\CC$, and $\alpha$ is an  arrow $\alpha : A \rarr FA$. We say that $A$ is the \emph{carrier} of the coalgebra, while $\alpha$ is the \emph{behaviour map}.

An \emph{$F$-coalgebra homomorphism} from $(A, \alpha)$ to $(B, \beta)$ is an arrow $h : A \rarr B$ such that
\begin{diagram}
A & \rTo^{\alpha} & FA \\
\dTo^{h} & & \dTo_{Fh} \\
B & \rTo_{\beta} & FB
\end{diagram}

$F$-coalgebras and their homomorphisms form a category $\FCoalg$.

An $F$-coalgebra $(C, \gamma)$ is \emph{final} if for every $F$-coalgebra $(A, \alpha)$ there is a unique homomorphism from $(A, \alpha)$ to $(C, \gamma)$, \ie if it is the terminal object in $\FCoalg$.

\begin{proposition}
If a final $F$-coalgebra exists, it is unique up to isomorphism.
\end{proposition}

\begin{proposition}[Lambek Lemma]
If $\gamma : C \rarr FC$ is final, it is an isomorphism
\end{proposition}


\subsection{Chu Spaces}

Chu spaces are a special case of a construction which originally appeared in \cite{Chu79},  written by Po-Hsiang Chu as an appendix to Michael Barr's monograph on $*$-autonomous categories \cite{Barr79}.

Chu spaces have several interesting aspects:
\begin{itemize}
\item They have a rich type structure, and in particular form models of Linear Logic \cite{Gir87,See89}.
\item They have a rich representation theory; many concrete categories of interest  can be fully embedded into Chu spaces \cite{LS91,DBLP:conf/lics/Pratt95}.
\item There is a natural notion of `local logic' on Chu spaces \cite{Barwise97}, and an interesting characterization of information transfer across Chu morphisms \cite{DBLP:journals/igpl/Benthem00}.
\end{itemize}
Applications of Chu spaces have been proposed in a number of areas, including concurrency \cite{DBLP:journals/mscs/Pratt03},
hardware verification \cite{DBLP:conf/cdes/Ivanov08}, game theory \cite{RePEc:usi:wpaper:417}  and fuzzy systems \cite{Pap00,DBLP:journals/jaciii/NguyenNWK01}. Mathematical studies concerning the general Chu construction include \cite{DBLP:journals/mscs/Pavlovic97,Barr98,DBLP:journals/acs/GiuliT07}.

We briefly review the basic definitions.

Fix a set $K$. A Chu space over $K$ is a structure $(X, A, e)$, where $X$ is a set of `points' or `objects', $A$ is a set of `attributes', and $e : X \times A \rarr K$ is an evaluation function.

A morphism of Chu spaces
\[ f : (X, A, e) \rarr (X', A', e') \]
is a pair of functions
\[ f = (f_{*} : X \rarr X', f^{*} : A' \rarr A) \]
such that, for all $x \in X$ and $a' \in A'$:
\[ e(x, f^{*}(a')) = e'(f_{*}(x), a') . \]
Chu morphisms compose componentwise: if $f : (X_{1}, A_{1}, e_{1}) \rarr (X_{2}, A_{2}, e_{2})$ and $g : (X_{2}, A_{2}, e_{2}) \rarr (X_{3}, A_{3}, e_{3})$, then
\[ (g \circ f)_{*} = g_{*} \circ f_{*}, \qquad (g \circ f)^{*} = f^{*} \circ g^{*} . \]
Chu spaces over $K$ and their morphisms form a category $\Chu_{K}$.

\subsection{Representing Physical Systems}

Our basic paradigm for representing physical systems, as laid out in \cite{btm}, is as follows.
We take a system to be specified by its set of \emph{states} $S$, and the set of \emph{questions} $Q$ which can be `asked' of the system. We shall consider only `yes/no' questions; however, the result of asking a question in a given state will in general be \emph{probabilistic}. This will be represented by an evaluation function
\[ e : S \times Q \rarr [0, 1] \]
where $e(s, q)$ is the probability that the question $q$ will receive the answer `yes' when the system is in state $s$. 
Thus a system is represented directly as a Chu space.

In particular, a quantum system with a Hilbert space $\HH$ as its state space will be represented as
\[ (\Ho, \LL(\HH), e_{\HH}) \]
where $\Ho$ is the set of non-zero vectors of $\HH$, $\LL(\HH)$ is the set of closed subspaces of $\HH$, and the evaluation function $\eH$ is the basic `statistical algorithm' of Quantum Mechanics:
\[ \eH(\psi, S) = \frac{\ip{\psi}{P_S \psi}}{\ip{\psi}{\psi}} = \frac{\ip{P_S \psi}{P_S \psi}}{\ip{\psi}{\psi}}  = \frac{\norm{P_S \psi}^2}{\norm{\psi}^2} .  \]
For a more detailed discussion see \cite{btm}. That paper goes on to show that:
\begin{itemize}
\item The biextensional collapse of this Chu space yields the usual projective representation of states as rays.
\item The Chu morphisms between these spaces are exactly the unitaries and unitaries, yielding a full and faithful functor from the groupoid of physical symmetries on Hilbert spaces to Chu spaces.
\item This representation is preserved by collapsing the unit interval to three values, but \emph{not} by the further collapse by either of the standard `possibilistic' reductions to two values.
\end{itemize}
This yields quite a pleasant picture. We would now like to investigate to what extent we can use coalgebras as an alternative setting for such representations; what problems arise, and on the other hand, what new possibilities become available.

\section{Comparison: A First Attempt}

We shall begin by showing that a subcategory of Chu spaces can be captured in completely equivalent form as a category of coalgebras.

Fix a set $K$. We can define a functor on $\Set$:
\[ F_K : X \mapsto K^{\pow X} . \]
If we use the \emph{contravariant} powerset functor, $F_K$ will be covariant. Explicitly, for $f : X \rarr Y$:
\[ F_Kf(g)(S) = g(f^{-1}(S)) , \]
where $g \in K^{\pow X}$ and $S \in \pow Y$.
A coalgebra for this functor will be a map of the form
\[ \alpha : X \rarr K^{\pow X} . \]
Consider a Chu space $C = (X, A, e)$ over $K$. We suppose furthermore that this Chu space is \emph{normal} (cf.~\cite{DBLP:conf/calco/PalmigianoV07} for a related but not identical use of this term), meaning that $A = \pow X$. 
Given this normal Chu space, we can define an $F_K$-coalgebra on $X$ by
\[ \alpha(x)(S) = e(x, S) . \]
We write $GC = (X, \alpha)$.

A coalgebra homomorphism from $(X, \alpha)$ to $(Y, \beta)$ is a function $h : X \rarr Y$ such that

\begin{diagram}
X & \rTo^{\alpha} & K^{\pow X} \\
\dTo<{h} & & \dTo>{Fh} \\
Y & \rTo_{\beta} & K^{\pow Y} \\
\end{diagram}

\begin{proposition}
\label{coalgchuprop}
Suppose we are given a Chu morphism $f : C \rarr C'$, where $C$ and $C'$ are normal Chu spaces, such that $f^* = f_*^{-1}$.
Then $f_* : GC \rarr GC'$ is an $F_{K}$-algebra homomorphism.
Conversely, given any $F_K$-algebra homomorphism $f : GC \rarr GC'$, then $(f, f^{-1}) : C \rarr C'$ is a Chu morphism.
\end{proposition}
\begin{proof}
Let $(f_*, f_*^{-1}) : C \rarr C'$ be a Chu space morphism.
Then
\[ \begin{array}{rcl}
(Ff_* \circ \alpha)(x)(S) & = & Ff_*(\alpha(x))(S) \\
& = & \alpha(x)(f_*^{-1} S) \\
& = & e(x, f_*^{-1}S) \\
& = & e(x, f^{*}S) \\
& = & e'(f_*(x), S) \\
& = & \beta \circ f_*(x)(S)
\end{array}
\]
so $f_*$ is a $F_K$-coalgebra homomorphism.
The converse is verified similarly (in fact by a cyclic permutation of the steps of the above proof).
\end{proof}
Let $\NChu_K$ be the category of normal Chu spaces and Chu morphisms of the form $(f, f^{-1})$.
Then by the Proposition, $G$ extends to a functor $G : \NChu_K \rarr \FCo$, with $G(f, f^{-1}) = f$.
Conversely, given an $F$-coalgebra $(X, \alpha)$, we can define a normal Chu space $H (X, \alpha) = (X, \pow X, e)$, where
$e(x, S) = \alpha(x)(S)$,
and given a coalgebra homomorphism $f : (X, \alpha) \rarr (Y, \beta)$, 
\[ H f = (f, f^{-1}) : H (X, \alpha) \rarr (Y, \beta) \]
will be a Chu morphism; this is verified in entirely similar fashion to Proposition~\ref{coalgchuprop}.

Altogether, we have shown:
\begin{theorem}
\label{isothm}
$\NChu_K$ and $\FCo$ are isomorphic categories, with the isomorphism witnessed by $G$ and $H = G^{-1}$.
\end{theorem}

\subsection{Discussion}

\subsubsection{A Critique of Coalgebras}
\label{critcosubsec}

\paragraph{Normality} Of course, the assumption of normality for Chu spaces is very strong; although it is worth mentioning that we have assumed nothing about either the value set or the evaluation function, in contrast to the notion of normality used in \cite{DBLP:conf/calco/PalmigianoV07} (for quite different purposes), which allows the attributes to be any subset of the powerset, but stipulates that $K = \Two$ and that the evaluation function is the characteristic function for set membership. One would like to extend the above correspondence to allow for wider classes of Chu spaces, in which the attributes need not be the full powerset. This is probably best done in an enriched setting of some kind.

It should also be said that the use of powersets, full or not,  to represent `questions' is fairly crude and ad hoc. The degree of freedom afforded by Chu spaces to choose \emph{both} the states \emph{and} the questions appropriately is a major benefit to conceptually natural and formally adequate modelling of a wide range of situations.

\paragraph{The Type Functor}
The experienced coalgebraist will be aware that the functors $F_K$ are problematic from the point of view of coalgebra. In particular, they fail to preserve weak pullbacks, and hence $\FCo$ will lack some of the nice structural properties one would like a category of coalgebras to possess. In fact, $F_K$ is a close cousin of the `double contravariant powerset', which is a standard counter-example for these properties \cite{DBLP:journals/tcs/Rutten00}.
However, much coalgebra can be done without this property \cite{GS05}, and recent work has achieved interesting results for coalgebras over the double contravariant powerset \cite{HKP07}.

A secondary problem is that as it stands, $\FCo$ cannot have a final coalgebra, for mere cardinality reasons. In fact, this issue can be addressed in a standard way. We can replace the contravariant powerset by a bounded version $\pow_{\kappa}$. We can also replace the function space by the \emph{partial function space} $\Pfn(X, Y)$. Thinking of partial functions in terms of their graphs, there is a set inclusion $\Pfn(X,Y) \subseteq \pow (X \times Y)$. Hence we can use a bounded version of the partial function functor, say $\Pfn_{\lambda}(X, Y)$, yielding those partial functions whose graphs have cardinality $< \lambda$. The resulting modified version of $F_K$:
\[ X \mapsto \Pfn_{\lambda}(\pow_{\kappa}(X), K) \]
is bounded, and admits a final coalgebra. Moreover, by choosing $\kappa$ and $\lambda$ sufficiently large, we can still represent a large class of systems whose behaviour involves total functions.

\paragraph{Behaviours vs. Symmetries}
However, there is a deeper conceptual problem which militates against the use of coalgebras in our context. An important property of physical theories is that they have rich symmetry groups (and groupoids), in which the key invariants are found, and from which the dynamics can be extracted. The main result of \cite{btm} was to recover these symmetries in the case of quantum systems as Chu morphisms.
The picture in coalgebra is rather different. One is concerned with behavioural or observational equivalence, as encapsulated by bisimulation, and the final coalgebra gives a `fully abstract' model of behaviour, in which bisimulation turns into equality. Moreover, every coalgebra morphism  is a functional bisimulation. If we consider the class of \emph{strongly extensional} coalgebras \cite{DBLP:journals/tcs/Rutten00}, those which have been quotiented out by bisimulation, they form a \emph{preorder}, and essentially correspond to the subcolagebras of the final coalgebra. Thus in a sense coalgebras are oriented towards maximum rigidity, and minimum symmetry. 

From this point of view, it would seem more desirable to have a \emph{universal homogeneous model},  with a maximum degree of symmetry, as a universal model for a large class of physical systems, rather than a final coalgebra.  Such a model has been constructed for bifinite Chu spaces in \cite{DBLP:conf/calco/DrosteZ07}. That context is too limited for our purposes here. It remains to be seen if universal homogeneous models can be constructed for larger subcategories of Chu spaces, encompassing those involved in our representation results.

In the present paper, we shall develop an alternative resolution of this problem by using  a fibred category of coalgebras, in which there is sufficient scope for variation to allow for the representation of symmetries. We shall use this to lift the representation theorem of \cite{btm} from Chu spaces to coalgebras.

\subsubsection{In Praise of Coalgebras}

\begin{itemize}
\item The coalgebraic point of view can be described as \emph{state-based}, but in a way that emphasizes that the meaning of states lies in their \emph{observable behaviour}. Indeed, in the ``universal model'' we shall construct, the states are determined exactly as the possible observable behaviours --- we actually find a canonical solution for \emph{what the state space should be} in these terms. States are identified exactly if they have the same observable behaviour.

We can see this as a kind of reconciliation between the \emph{ontic} and \emph{epistemic} standpoints, in which moreover \emph{operational ideas} are to the fore.

\item Coalgebras allow us to capture the `dynamics of measurement' --- what happens \emph{after a measurement} --- in a way that Chu spaces don't. They have \emph{extension in time} \cite{DBLP:conf/nato/AbramskyGN96}. We explain what we mean by this in more detail below.

\end{itemize}

\paragraph{Extension in Time}
Consider a coalgebraic representation of \emph{stochastic transducers}:
\[ F : X \mapsto \Prob (O \times X)^I \]
where $I$ is a fixed set of \emph{inputs}, $O$ a fixed set of \emph{outputs}, and $\Prob (S)$ is the set of probability distributions of finite support on $S$. This expresses the behaviour of a state $x \in X$ in terms of how it responds to an input $i \in I$ by producing an output $o \in O$ and evolving into a new state $x' \in X$. Since the automaton is stochastic, what is specified for each input $i$ is a probability distribution over the pairs $(o, x')$ comprising the possible responses.

We can think of $I$ as a set of \emph{questions}, and $O$ as a set of \emph{answers} (which we could standardize by only considering yes/no questions). Thus we can see such a stochastic automaton as a variant of the representation of physical systems we discussed previously, with the added feature of \emph{extension in time} --- the capacity to represent behaviour under repeated interactions.

What we can learn from this observation, incidentally, is that
\begin{center}
\fbox{QM is \emph{less} nondeterministic/probabilistic than stochastic transducers}
\end{center}
since in Quantum Mechanics, if we know the preparation and the outcome of the measurement, we know (by the projection postulate) exactly what the resulting quantum state will be. In automata theory, by contrast, even if we know the current state, the input, and which observable output was produced in response, we still do not  know in general what the next state will be. Could there be physical theories of this type?

\section{Semantics In One Country}

As a first step to developing a viable coalgebraic approach to representing physical systems, we shall hold a single system fixed, and see how we can represent this coalgebraically. This simple step eliminates most of the problems with coalgebras which we encountered in the previous Section. We will then have to see how variation of the system being represented can be reintroduced.

\subsection{Coalgebraic Semantics For One System}

We fix attention on a single Hilbert space $\HH$. This determines a set of questions $Q = \LL(\HH)$.
We now define an endofunctor on $\Set$:
\[ F^Q : X \mapsto (\One + (0, 1] \times X)^Q . \]
A coalgebra for this functor is then  a map
\[ \alpha : X \rarr (\One + (0, 1] \times X)^Q \]
The interpretation is that $X$ is a set of states; the coalgebra map sends a state to its behaviour, which is a function from questions in $Q$ to the probability that the answer is `yes'; and, \emph{if the probability is not 0}, to the successor state following a `yes' answer.

Unlike the functors $F_K$, the functors $F^Q$ are very well-behaved from the point of view of coalgebra (they are in fact \emph{polynomial functors} \cite{DBLP:journals/tcs/Rutten00}). They preserve weak pull-backs, which guarantees a number of nice properties, and they are bounded and admit final coalgebras
\[ \gQ : \UU_Q \rarr (\One + (0, 1] \times \UU_Q)^Q . \]
The elements of $\UU_Q$ can be visualized as `$Q$-branching trees',  with the arcs labelled by probabilities.

The $F^Q$-coalgebra which is of primary interest to us is 
\[ \aH : \Ho \rarr (\One + (0, 1] \times \Ho)^Q \]
defined by:
\[ \aH(\psi)(S) = \left\{ \begin{array}{ll}
 0, & \eH(\psi, S) = 0 \\
 (r, P_S \psi), & \eH(\psi, S) = r > 0.
 \end{array}
 \right.
 \]
The new ingredient compared with the Chu space representation of $\HH$ is the state which results in the case of a `yes' answer to the question, which is computed according to the (unnormalized) \emph{L\"uders rule}.
 
This system will of course have a representation in the final coalgebra $(\UU_Q, \gamma_Q)$, specified by the unique coalgebra homomorphism $h : (\Ho, \aH) \rarr (\UU_Q, \gamma_Q)$.

\section{Indexed Structure For Coalgebras}

Our strategy will now be to \emph{externalize contravariance as indexing}. This will allow us to alleviate many of the problems we  encountered with using coalgebras to represent physical systems, and to access the power of the coalgebraic framework. In particular, we will be able to construct a \emph{single universal model for quantum systems}.

We shall define a functor
\[ \IndF : \Set^{\mathsf{op}} \rarr \CAT \]
where $\CAT$ is the `superlarge'\footnote{For those concerned with set-theoretic foundations, we shall on a couple of occasions refer to `superlarge' categories such as $\CAT$, the category of `large categories'  such as $\Set$. If we think of large categories as based on classes, superlarge categories are based on entities `one size up' --- `conglomerates' in the terminology of \cite{HS73}. This can be formalized in set theory with a couple of Grothendieck universes.} category of categories and functors.
$\IndF$ is defined on objects by
\[ Q \mapsto F^Q{-}\Coalg. \]
For  a function $f : Q' \rarr Q$, we define
\[ t^f_X : F^Q(X) \rarr F^{Q'}(X) :: \Theta \mapsto \Theta \circ f \]
and 
\[ \IndF(f) = f^* : F^Q{-}\Coalg \rarr F^{Q'}{-}\Coalg \]
\[ f^* : (X, \alpha) \mapsto (X, t^f_X \circ \alpha), \qquad f^* : (h : (X, \alpha) \rarr (Y, \beta))  \mapsto h . \]

\begin{proposition}
\label{mtfprop}
For each $f : Q' \rarr Q$, $t^f$ is a natural transformation, and $f^*$ is a functor.
\end{proposition}
\begin{proof}
The naturality of $t^f$ is the diagram
\begin{diagram}
X && F^Q X & \rTo^{t^f_X} & F^{Q'} X \\
\dTo<{g} && \dTo<{F^Q g} & & \dTo>{F^{Q'} g} \\
Y && F^Q Y & \rTo_{t^f_Y} & F^{Q'} Y
\end{diagram}
This diagram commutes because $t^f$ acts by pre-composition and $F^Q$, $F^{Q'}$ by post-composition.
For any $\Theta \in F^Q X$, we obtain the common value
\[ (1 + (1 \times g)) \circ \Theta \circ f . \]

It is a general fact \cite{DBLP:journals/tcs/Rutten00} that a natural transformation $t : F \rarr G$ induces a functor between the coalgebra categories in the manner specified above. The fact that the coalgebra homomorphism condition is preserved follows from the commutativity of
\begin{diagram}
X & \rTo^{\alpha} & FX & \rTo^{t_X} & GX \\
\dTo<{h} & & \dTo>{Fh} & & \dTo>{Gh} \\
Y & \rTo_{\beta} & FY & \rTo_{t_Y} & GY
\end{diagram}
The left hand square commutes because $h$ is an $F$-coalgebra homomorphism; the right hand square is naturality of $t$.
\end{proof}

Thus we get a \emph{strict indexed category} of coalgebra categories, with contravariant indexing.

\subsection{The Grothendieck Construction}
We now recall an important general construction.
Where we have an indexed category, we can apply the \emph{Grothendieck construction} \cite{Groth70}, to glue all the fibres together (and get a fibration).

Given a functor
\[ \II : \CC^{\mathsf{op}} \rarr \CAT \]
we define $\int \II$ with objects $(A, a)$, where $A$ is an object of $\CC$ and $a$ is an object of $\II(A)$.
Arrows are $(G, g) : (A, a) \rarr (B, b)$, where $G : B \rarr A$ and $g: \II(G)(a) \rarr b$.

Composition of $(G, g) : (A, a) \rarr (B, b)$ and $(H, h) : (B, b) \rarr (C, c)$ is given by
\[ (G \circ H, h \circ \II(H)(g)) : (A, a) \rarr (C, c) . \]

Applying the Grothendieck construction to $\IndF$, we can now put all our categories of coalgebras, indexed by the sets of questions,  together in one category. We will use this to get our universal model for quantum systems.

Before turning to this, we will consider an analogous indexed structure for Chu spaces, which will allow us to define a comparison functor between the two models.

\section{Indexed Comparison With Chu Spaces}

\subsection{Slicing and Dicing Chu}

For each $Q$, we define $\Chu_K^Q$ to be the subcategory of $\Chu_K$ of Chu spaces $(X, Q, e)$ and morphisms of the form $(f_*, \id_{Q})$.

This doesn't look too exciting. In fact, it is just the comma category
\[ ({-} \times Q, \hat{K}) \]
where $\hat{K} : \mathbf{1} \rarr \Set$ picks out the object $K$.

Given $f : Q' \rarr Q$, we define a functor
\[ f^* : \Chu^Q_K \rarr \Chu^{Q'}_K :: (X, Q, e) \mapsto (X, Q', e \circ (1 \times f)) \]
and which is the identity on morphisms.
To verify functoriality, we only need to check that the Chu morphism condition is preserved.
That is, we must show, for any morphism $(f_*, \id_{Q}) : (X, Q, e)  \rarr (X', Q, e')$, $x \in X$, and $q' \in Q$, that
\[ e(x, f(q')) = e'(f_*(x), f(q')) \]
which follows from the Chu morphism condition on $(f_*, \id_{Q})$.

This gives an indexed category
\[ \CHU_K : \Set^{\mathsf{op}} \rarr \CAT . \]

\subsection{Grothendieck puts Chu back together again}

The fibre categories $\Chu^Q_K$ are pale reflections of the full category of Chu spaces, trivialising the contravariant component of morphisms. However, the Grothendieck construction gives us back the full category.

\begin{proposition}
\label{Chuisoprop}
\[ \int \CHU_K \cong \Chu_K.  \]
\end{proposition}
\begin{proof}
Expanding the definitions, we see that objects in $\int \CHU_K$ have the form
\[ (Q, (X, Q, e : X \times Q \rarr K)) \]
while morphisms have the form
\[ (f, (f_*, \id_{Q'})) : (Q, (X, Q, e)) \rarr (Q', (X, Q', e')) \]
where $f : Q' \rarr Q$, and 
\[ (f_*, \id_{Q'}) : (X, Q', e \circ (1 \times f)) \rarr (X', Q', e') \]
is a  morphism  in $\Chu^{Q'}_K$. The morphism condition is:
\[ e(x, f(q')) = e'(f_*(x), q') . \]
This is exactly the Chu morphism condition for
\[ (f_*, f) : (X, Q, e) \rarr (X', Q', e') . \]
Composition of $(f, (f_*, \id_{Q'}))$ with $(g, (g_*, \id_{Q''}))$ is given by $(f \circ g, (g_* \circ f_*, \id_{Q''}))$.

The isomorphism with $\Chu_K$ is immediate from this description.
\end{proof}

\subsection{The Truncation Functor}
\label{truncsubsec}

The relationship between coalgebras and Chu spaces is further clarified by an indexed \emph{truncation functor} $\TT : \IndF \rarr \CHU$.

For each set $Q$ there is a functor
\[ \TT^Q : F^Q{-}\Coalg \rarr \Chu^Q_K \]
This is defined on objects by
\[ \TT^Q(X, \alpha) = (X, Q, e) \]
where
\[ e(x, q) = \left\{ \begin{array}{ll}
0, & \alpha(x)(q) = 0 \\
r, & \alpha(x)(q) = (r, x')
\end{array}
\right.
\]
The action on morphisms is trivial:
\[ \TT^Q : (h : (X, \alpha) \rarr (Y, \beta)) \mapsto (h, \id_{Q}) . \]
The verification that coalgebra homomorphisms are taken to Chu morphisms is straightforward. The fact that each $\TT^Q$ is a faithful functor is then immediate.

For each $f : Q' \rarr Q$, we have the naturality square
\begin{diagram}[width=5em]
F^Q{-}\Coalg & \rTo^{\TT^Q} & \Chu^Q_K \\
\dTo<{\IndF(f)} & & \dTo>{\CHU_K(f)} \\
F^{Q'}{-}\Coalg & \rTo_{\TT^{Q'}} & \Chu^{Q'}_K
\end{diagram}
On objects, both paths around the diagram carry a coalgebra $(X, \alpha)$ to the Chu space $(X, Q', e)$, where
\[ e(x, q') = \left\{ \begin{array}{ll}
0, & \alpha(x)(f(q')) = 0 \\
r, & \alpha(x)(f(q')) = (r, x')
\end{array}
\right.
\]
The action on morphisms in both cases is trivial: a coalgebra homomorphism $h$ is sent to the Chu morphism $(h, \id_{Q'})$.

We can summarize this as follows:
\begin{proposition}
$\TT : \IndF \lrarr \CHU$ is a strict indexed functor, which is faithful on each fibre.
\end{proposition}

As an immediate corollary, we obtain:

\begin{proposition}
\label{faithindchuprop}
There is a faithful functor $\int \TT : \int \IndF \lrarr \int \CHU \; \cong \; \Chu_K$.
\end{proposition}

We can also refine the isomorphism of Theorem~\ref{isothm}. We say that an $F^Q$-coalgebra $(X, \alpha)$ is \emph{static} if for all $x \in X$:
\[ \alpha(x)(q) = (r, x') \Imp x' = x . \]
Thus in a static coalgebra, observing an answer to a question has no effect on the state.
We write $S^Q{-}\Coalg$ for the full subcategory of $F^Q{-}\Coalg$ determined by the static coalgebras. 
This extends to an indexed subcategory $\IndS$ of $\IndF$, since the functors $f^*$, for $f : Q' \rarr Q$,  carry $S^Q{-}\Coalg$ into $S^{Q'}{-}\Coalg$.

\begin{proposition}
For each set $Q$, $\Chu^Q_K$ is isomorphic to $S^Q{-}\Coalg$. Moreover this is an isomorphism of strict indexed categories.
\end{proposition}
\begin{proof}
We can define an indexed functor 
\[ \EE^Q : \Chu^Q_K \rarr S^Q{-}\Coalg \]
\[ \EE^Q : (X, Q, e) \mapsto (X, \alpha) \]
where 
\[ \alpha(x)(q) = \left\{ \begin{array}{ll}
0, & e(x, q) = 0 \\
(r, x), & e(x, q) = r > 0.
\end{array}
\right.
\]
$\EE^Q$ takes a Chu morphism $(f, \id_{Q})$ to $f$.

It is straightforward to verify that this is an indexed functor, and inverse to the restriction of $\TT$ to $\IndS$.
\end{proof}
We can combine this with Proposition~\ref{Chuisoprop} to obtain:
\begin{theorem}
The category of Chu spaces $\Chu_K$ is isomorphic to a full subcategory of $\int \IndF$, the Grothendieck category of an indexed category of coalgebras.
\end{theorem}

This gives a  clear picture of how coalgebras extend Chu spaces with some `observational dynamics'.

\section{A Universal Model}

We can now define a single coalgebra which is \emph{universal for quantum systems}.

We proceed in  a number of steps:
\begin{enumerate}
\item Fix a countably-infinite-dimensional Hilbert space, e.g. $\HH_{\UU} = \ell_2(\Nat)$, with its standard orthonormal basis $\{ e_{n} \}_{n \in \Nat}$. Take $Q = \LL(\HH_{\UU})$.
Let $\Ug$ be the final coalgebra for $F^Q$.
\item Any quantum system is described by a separable Hilbert space $\KK$. In practice, the Hilbert space chosen to represent a given system will come with a preferred orthonormal basis $\{ \psi_{n} \}$. This basis will induce an isometric embedding 
\[ i : \KK \rmon \HH_{\UU} \;\; :: \;\; \psi_{n} \longmapsto e_{n} . \]
Taking $Q' = \LL(\KK)$, this induces a map $f = i^{-1} : Q \rarr Q'$. This in turn induces a functor $f^* : F^{Q'}{-}\Coalg \rarr F^Q{-}\Coalg$.
\item This functor can be applied to the coalgebra $(\Ko, \aK)$  corresponding to the Hilbert space $\KK$ to yield a coalgebra in $F^Q{-}\Coalg$.
\item Since $\Ug$ is the final coalgebra in $F^Q{-}\Coalg$, there is a unique coalgebra homomorphism $\fsem{\cdot}_{\Ko} : f^*(\Ko, \aK) \rarr \Ug$.
\item This homomorphism maps the quantum system $(\Ko, \aK)$ into $\Ug$ in a \emph{fully abstract fashion}, \ie identifying states precisely according to observational equivalence.
\item This homomorphism is an arrow in the Grothendieck category $\int \IndF$.
\item This works for \emph{all} quantum systems, with respect to a \emph{single} final coalgebra. 
\end{enumerate}

This is  a `Big Toy Model' in the sense of \cite{btm}.

We shall now investigate the nature of this coalgebraic semantics for physical systems in more detail.

\subsection{Bisimilarity and Projectivity}

Our first aim is to characterize when two states of a physical system are sent to the same element of  the final coalgebra by the semantic map $\fsem{\cdot}$.
We can call on some general coalgebraic notions for this purpose.

We shall begin with one of the key ideas in the theory of coalgebra, \emph{bisimilarity}. This can be defined in generality for coalgebras over any endofunctor \cite{DBLP:journals/tcs/Rutten00}, but we shall just give the concrete definition as it pertains to $F^Q{-}\Coalg$. Given $F^Q$-coalgebras $(X, \alpha)$ and $(Y, \beta)$, a \emph{bisimulation} is a relation $R \subseteq X \times Y$ such that:
\[ \begin{array}{rccl}
 x R y & \Imp & \forall q \in Q. & \alpha(x)(q) = 0 \Imp \beta(y)(q) = 0 \\ 
 & & \wedge &
 \alpha(x)(q) = (r, x') \Imp \beta(y)(q)  = (r, y') \; \wedge \;  x' R y' .
\end{array}
\]
We say that $x$ and $y$ are \emph{bisimilar}, and write $x \bis y$, if there is some bisimulation $R$ with $xRy$. Note that bisimilarity can hold between elements of \emph{different} coalgebras. This means that states of different systems can be compared in terms of a common notion of observable behaviour.

The above definition is given in an apparently asymmetric form, but $\bis$ is easily seen to be a symmetric relation, since the cases $\alpha(x)(q) = 0$ and $\alpha(x)(q) = (r, x')$ are mutually exclusive and exhaustive.

\begin{proposition}
Bisimilarity is an equivalence relation.
\end{proposition}
\begin{proof}
The main point is transitivity, which follows automatically since the polynomial functor $F^Q$ preserves pullbacks \cite{DBLP:journals/tcs/Rutten00}.
\end{proof}

The key feature of bisimilarity is given by the following proposition, which is also standard for functors preserving weak pullbacks \cite{DBLP:journals/tcs/Rutten00}. We consider coalgebras for  such a functor $F$ for which a final coalgebra exists. Given an $F$-coalgebra $(X, \alpha)$ and $x \in X$, we write $\fsem{x}$ for the denotation of $x$ in the final coalgebra.
\begin{proposition}
For any $F$-coalgebras $(X, \alpha)$ and $(Y, \beta)$, and $x \in X$, $y \in Y$:
\[ \fsem{x} = \fsem{y} \IFF x \bis y . \]
\end{proposition}
Thus bisimilarity characterizes equality of denotation in the final coalgebra semantics.

We begin by characterizing bisimilarity in the coalgebra $(\Ko, \aK)$ arising from the Hilbert space $\KK$, for the functor $F^Q$, where $Q = \LL(\KK)$.

We define the usual \emph{projective equivalence} on the non-zero vectors of a Hilbert space $\Ko$ by:
\[ \psi \proj \phi \IFF \exists \lambda \in \Complex . \, \psi = \lambda \phi . \]
Thus two vectors are projectively equivalent if they belong to the same \emph{ray} or one-dimensional subspace.

\begin{proposition}
\label{projbisprop}
For any vectors $\psi, \phi \in \Ko$:
\[ \psi \proj \phi \IFF \psi \bis \phi . \]
\end{proposition}
\begin{proof}
Firstly, recall the definition of $\eK$ from Section~2.3. We can describe the bisimilarity condition on a relation $R \subseteq \Ko^2$ for the coalgebra $(\Ko, \aK)$ more directly as follows:
\[ \psi \, R \, \phi \Imp \forall S \in \LL(\HH). \, \eK(\psi, S) = \eK(\phi, S) \; \wedge \; (P_S \psi) \, R \, (P_S \phi) . \]
Thus if $\psi \bis \phi$, then for all $S \in \LL(\KK)$, $\eK(\psi, S) = \eK(\phi, S)$, and hence $\psi \proj \phi$ by Proposition~3.2 of \cite{btm}.
For the converse, it suffices to show that the relation ${\proj} \subseteq \Ko^2$ is a bisimulation.
If $\psi = \lambda \phi$, then for all $S$, $\eK(\psi, S) = \eK(\phi, S)$ by Proposition~3.2 of \cite{btm}, and $P_S\psi = \lambda P_S \phi$, so $\proj$ is a bisimulation as required.
\end{proof}

We now show that bisimilarity in Hilbert spaces is stable under transport across fibres by isometries.

Firstly, we have a general property of fibred coalgebras.
\begin{proposition}
\label{surjprop}
If $f : Q' \rarr Q$ is surjective, then bisimulation on the $F^{Q'}$-coalgebra $f^{*}(X, \alpha)$ coincides with bisimulation on the $F^{Q}$-coalgebra $(X, \alpha)$.
\end{proposition}
\begin{proof}
Unwinding the definitions of the two bisimulation conditions on relations, the only difference is that one quantifies over questions $q \in Q$, and the other over questions $f(q')$, for $q' \in Q'$. If $f$ is surjective, these are equivalent.
\end{proof}

Given a Hilbert space $\KK$ and an isometric embedding $i : \KK \rmon \HH_{\UU}$, let  $Q = \LL(\HH_{\UU})$, $Q' = \LL(\KK)$, $f = i^{-1} : Q \rarr Q'$. Then the $F^Q$-coalgebra $f^*(\Ko, \aK)$ is $(\Ko, \beta)$, where:
\[ \beta(\psi)(S) = \aK(\psi)(i^{-1}(S)) . \]
\begin{proposition}
Bisimulation on the elements of the $F^Q$-coalgebra $(\Ko, \beta)$ coincides with bisimulation on the $F^{Q'}$-coalgebra $(\Ko, \aK)$. If we identify $\KK$ with the subspace $\HH' \rinc \HH_{\UU}$ determined by the image of $i$, it also coincides with bisimulation on $\HH'$. It is also the restriction of bisimulation on $\HH_{\UU}$.
\end{proposition}
\begin{proof}
Since $i$ is an isometry, the direct image $i(S)$ of a closed subspace of $\KK$ is a closed subspace of $\HH_{\UU}$, and since $i$ is injective, $i^{-1}(i(S)) = S$. Thus $i^{-1}$ is surjective, yielding the first statement by Proposition~\ref{surjprop}. The fact that $i$ is an isometric embedding also guarantees that $\eK(\psi, S) = e_{H_{\UU}}(\psi, S)$ for $\psi \in \HH'$, $S \in \LL(\HH')$. Finally, by Proposition~\ref{projbisprop}, bisimulation  on Hilbert spaces coincides with projective equivalence, and projective equivalence on $\HH'$ is the restriction of projective equivalence on $\HH_{\UU}$.
\end{proof}

Putting these results together, we have the following:
\begin{theorem}
\label{extqthm}
Let $\fsem{\cdot}_{\Ko} : f^*(\Ko, \aK) \rarr \Ug$ be the final coalgebra semantics for $\Ko$ with respect to the isometric embedding $i : \KK \rmon \HH_{\UU}$. Then for any $\psi, \phi \in \Ko$:
\[ \fsem{\psi}_{\Ko} = \fsem{\phi}_{\Ko} \IFF \psi \bis \phi \IFF \psi \proj \phi . \]
\end{theorem}
Thus the \emph{strongly extensional quotient} \cite{DBLP:journals/tcs/Rutten00} of the coalgebra $(\Ko, \aK)$ is the \emph{projective coalgebra} $(\PP (\KK), \baK)$, where $\PP(\KK)$ is the set of rays or one-dimensional subspaces of $\KK$, and $\baK$ is defined by:
\[ \baK(\ray{\psi}) = \left\{ \begin{array}{ll}
0, & \alpha(\psi) = 0 \\
(r, \ray{\phi}) & \alpha(\psi) = (r, \phi).
\end{array}
\right.
\]
Here $\ray{\psi} = \{ \lambda \psi \mid \lambda \in \Complex \}$ is the ray generated by $\psi$.

\paragraph{Remark} There is a subtlety lurking here, which is worthy of comment.
When we consider an extension of a Hilbert space to a larger one, $\HH' \rinc \HH$, the characteristic quantum phenomenon of \emph{incompatibility} can arise; a subspace $S$ of $\HH$ may be incompatible with the subspace $\HH'$ (so that e.g. the corresponding projectors do not commute). The characterization of bisimulation as projective equivalence shows that this notion is nevertheless stable under such extensions. However, we can expect incompatibility to be reflected in some fashion in the coalgebraic approach, in particular in the development of a suitable \emph{coalgebraic logic}.

\subsection{Representing Physical Symmetries}

We shall now show that the passage to the Grothendieck category of coalgebras does succeed in alleviating the problem of excessive rigidity of coalgebras as discussed in Section~\ref{critcosubsec}.
Our strategy will be to lift the Representation Theorem~3.15 from \cite{btm} from Chu spaces to coalgebras, using the results of Section~\ref{truncsubsec}.

We consider a morphism in $\int \IndF$ between representations of Hilbert spaces.
Such a morphism has the form
\[ h : f^*(\Ho, \aH) \rarr (\Ko, \aK) \]
where $\HH$ and $\KK$ are any Hilbert spaces, and writing $Q = \LL(\HH)$, $Q' = \LL(\KK)$, the functor $f^*$ is induced by a map $f : Q' \rarr Q$, and $h$ is a homomorphism of $F^{Q'}$-coalgebras.

By Proposition~\ref{faithindchuprop}, 
\[ (h, f) : (\Ho, \LL(\HH), \eH) \rarr (\Ko, \LL(\KK), \eK) \]
is a Chu morphism.
By Proposition~3.2 and the remark following Theorem 3.10 of \cite{btm}, the Chu morphism induced by the biextensional collapse of these Chu spaces is
\[ (\PP h, f) : (\PP(\Ho), \LL(\HH), \beH) \rarr (\PP(\Ko), \LL(\KK), \beK) \]
where $\PP(h)(\ray{\psi}) = \overline{h(\psi)}$. By Theorem~\ref{extqthm}, the induced coalgebra homomorphism on the strongly extensional quotients of the corresponding coalgebras is
\[ \PP h : f^*(\PP(\HH), \baH) \rarr (\PP(\KK), \baK) . \]
We can now use Theorem~3.12 of \cite{btm}:
\begin{theorem}
\label{mainth}
Let $\HH$, $\KK$ be Hilbert spaces of dimension greater than 2. Consider a Chu morphism
\[ (f_*, f^*) :   (\PP(\HH), \LL(\HH), \beH) \rarr (\PP(\KK), \LL(\KK), \beK) . \]
where $f_*$ is injective. Then there is a semiunitary (\ie a unitary or antiunitary) $U : \HH \rarr \KK$ such that $f_* = \PP(U)$. $U$ is unique up to a phase. Moreover, $f^{*}$ is then uniquely determined as $U^{-1}$.
\end{theorem}

Since any coalgebra homomorphism gives rise to a Chu morphism, this will allow us to lift fullness of the representation in Chu spaces to the coalgebraic setting.

\begin{proposition}
If $U : \HH \rarr \KK$ is a semiunitary, then $\Uo : f^{*}(\Ho, \aH) \rarr (\Ko, \aK)$ is a coalgebra homomorphism, where $f^{*} = U^{-1}$.
\end{proposition}
\begin{proof}
This follows by the same argument as Proposition~3.13 of \cite{btm}. In particular, the fact that $\Uo$ is a coalgebra homomorphism follows from the relation
\[ P_{S} (U \psi) = U(P_{U^{-1}(S)} \psi)  \]
which is shown there.
\end{proof}

We must now account for the injectivity hypothesis in Theorem~\ref{mainth}.
The following properties of coalgebras and Chu spaces respectively are standard.
\begin{proposition}
If $F$ preserves weak pullbacks, the kernel of an $F$-coalgebra homomorphism is a bisimulation. Hence
if $(A, \alpha)$ is a strongly extensional $F$-coalgebra, on which bisimilarity is equality, then any homomorphism with $(A, \alpha)$ as domain must be injective.
\end{proposition}

\begin{proposition}
If $f : C_{1} \rarr C_{2}$ is a morphism of separated Chu spaces, and $f^{*}$ is surjective, then $f_{*}$ is injective.
\end{proposition}
We shall write $\sind$ for the restriction of $\IndF$ to $\sSet$, the category of sets and surjective maps.
Similarly, we write $\schu$ for the restriction of $\CHU$ to $\sSet$. Clearly $\TT$ cuts down to these restrictions. Moreover, the isomorphism of $\Chu_{K}$ with $\int \CHU$ of Proposition~\ref{Chuisoprop} cuts down to an isomorphism of $\int \schu$ with $\Chus_{K}$, the subcategory of Chu spaces and morphisms $f$ with $f^{*}$ surjective.

Thus if we define the category $\PSymm$ as in \cite{btm}, with objects Hilbert spaces of dimension $> 2$, and morphisms semiunitaries quotiented by phases, we obtain the following result:

\begin{theorem}
There is a full and faithful functor $\PP C : \PSymm \rarr \int \sind$. Moreover, the following diagram commutes:
\begin{diagram}[heads=vee,width=5em]
\PSymm & \rrep^{\PP C} & \int \sind \\
\drep<{\PP R} & & \dmon>{\int \TT} \\
\Chus_{[0, 1]} & \lTo_{\cong} & \int \schu
\end{diagram}
Here $\PP R$ is the full and faithful functor of Theorem~3.15 of \cite{btm}.
\end{theorem}

This result confirms that our approach of expressing contravariance through indexing over a base does succeed in allowing sufficient scope for the representation of physical symmetries, while also allowing for the construction of a universal model as a final coalgebra, and for the expression of the dynamics of repeated measurements.

\section{Bivariant Coalgebra}

Our development of `coalgebra with contravariance' can be carried out quite generally. We shall briefly sketch this general development.

Suppose we have a functor
\[ G : \CC^{\mathsf{op}} \times \CC \lrarr \CC .  \]
Since $\CAT$ is cartesian closed, we can curry $G$ to obtain
\[ \hat{G} : \CC^{\mathsf{op}} \lrarr [\CC, \CC] \]
where $[\CC, \CC]$ is the (superlarge) functor category on $\CC$. There is also a functor
\[ [\CC, \CC] \lrarr \CAT \]
which sends a functor $F$ to its category of coalgebras, and a natural transformation $t : F \rarr G$ to the corresponding functor between the categories of coalgebras, as in Proposition~\ref{mtfprop}.
Composing these two functors, we obtain a strict indexed category
\[ \GG : \CC^{\mathsf{op}} \lrarr \CAT . \]
We can then form the Grothendieck category $\int \GG$.

The indexed category $\IndF$ arises in exactly this way, from the functor
\[ G : \Set^{\mathsf{op}} \times \Set \lrarr \Set \;\; :: \;\;(Q, X) \longmapsto (\One + (0, 1] \times X)^Q . \]

We have found this combination of fibrational and coalgebraic structure a convenient one for our objective in the present paper of representing physical systems. 
In particular, the fibrational approach to contravariance allows enough `elbow room' for the representation of symmetries. We also used the fibrational structure in formulating the connection to Chu spaces, which proved to be both technically useful and conceptually enlightening.
A natural follow-up would be to develop a fibred version of coalgebraic logic, which we plan to do in a sequel.

We note that a quite different, and in some sense more direct approach to coalgebra for bivariant functors has been developed by Tews \cite{DBLP:journals/entcs/Tews00}. A viable approach is developed in \cite{DBLP:journals/entcs/Tews00} only for a limited class of functors, the `extended polynomial functors'. Moreover, the issues of rigidity vs.~symmetry which we have been concerned with are not addressed in this approach, which is also technically fairly complex.
Of course, there is a beautiful theory of the solution of reflexive equations for mixed-variance functors provided by \emph{Domain theory} \cite{DBLP:books/el/leeuwen90/GunterS90,paper30}. The value of coalgebras, in our view, is that they provide a simpler setting in which a great deal can be very effectively accomplished, without the need for the introduction of partial elements and the like. 

The need for contravariance in our context, motivated by the representation of physical systems, appears to be of a different nature, and hence better met by the fibrational methods we have introduced in the present paper.

A deeper understanding of the issues here will, we hope, shed interesting light on each of the topics we have touched on in this paper: foundations of physics, computational models, and the mathematics of coalgebras.

\bibliographystyle{plain}

\bibliography{../bibfile}

\end{document}